\documentclass[preprint,12pt,authoryear]{elsarticle}

\usepackage{amsmath,amssymb,amsthm,mathtools}
\usepackage{microtype}
\usepackage{graphicx}
\usepackage{placeins}
\usepackage[mathlines]{lineno}
\usepackage{hyperref}

\journal{Mathematical Social Sciences}
\biboptions{round,authoryear}

\newcommand{\LCC}{L}

\newtheorem{theorem}{Theorem}
\newtheorem{proposition}{Proposition}
\newtheorem{corollary}{Corollary}
\newtheorem{lemma}{Lemma}

\begin{document}

\begin{frontmatter}

\title{Who Pays for a Connected Public Good?}

\author[aff1]{Marco Tulio Angulo}
\address[aff1]{Institute of Mathematics, Universidad Nacional Aut\'onoma de M\'exico, M\'exico}

\begin{abstract}
Connectivity can turn separate contributions into a public good that benefits
participants and nonparticipants alike. Yet only participants pay for it, and each can
avoid her charge by withdrawing while enjoying whatever public output survives. How
do network connections constrain who can or must pay, and how large a cost can
the group bear? We study a participation game in which public output equals the size of the largest connected group of active players, whose members divide a fixed operating cost. The output a participant destroys by withdrawing is her withdrawal responsibility and bounds her payment. These bounds determine the group's cost-bearing capacity and, whenever the group can be kept active, which participants must pay.
We prove that among connected networks with at least five players, the path uniquely maximizes the cost-bearing capacity of a group, yet minimizes the expected surviving output after a uniformly selected participant withdraws. As the group grows, the ratio of maximal capacity to gross social benefit under full participation converges to one quarter. Fragility can therefore strengthen voluntary finance, but even the most financeable network cannot close the gap between the value it creates and what its participants will pay to sustain it.
\end{abstract}


\end{frontmatter}

%\linenumbers

\section{Introduction}
\label{sec:introduction}

After an earthquake disables ordinary communications, local organizations restore service by operating emergency radio relays. The service benefits operators and nonoperators alike, but only active operators share its operating cost. An operator can avoid her charge by switching off her relay while continuing to use whatever communications system survives. The largest connected component is the largest set of active relays that can operate as one system, so we measure public output by its number of relays. The output destroyed by a withdrawal therefore depends on the operator’s network position. Figure~\ref{fig:fig1a} compares two configurations of eight active relays, a cycle and a path. Both initially produce eight units of output. Removing any relay from the cycle leaves the other seven connected, so output falls by one. Removing a central relay from the path leaves no connected component larger than four, so output falls by four. An operator will remain active only if her charge does not exceed the benefit she would lose by switching off her relay. The network therefore constrains which cost divisions can keep the group active and whether the entire bill can be covered. Who, then, can or must pay? What is the largest operating cost the group can finance?

The same financing problem arises when separate contributions create value
only through a connected system. Community broadband provides a close analogue. In guifi.net, participants pool resources to build and operate common network infrastructure and use an explicit mechanism to divide shared transit costs \citep{BaigEtAl2015,CerdaAlabernEtAl2020}. Private-land biodiversity conservation shares the productive importance of connectivity, though not the same financing institution. Conserving a parcel creates biodiversity benefits beyond its owner and may connect otherwise separate habitats \citep{ParkhurstEtAl2002,PolaskyEtAl2014}. Across these settings, contributors bear costs, benefits extend beyond them, and the value of each contribution depends partly on what it connects. Withdrawal may therefore remove only one contribution or fragment a larger system, while the withdrawing contributor continues to enjoy whatever public benefit survives.

We isolate this mechanism in a participation game on a fixed network. Every
player, whether active or inactive, benefits from the size of the largest connected group of active players, while all active players divide the same real operating cost. Before players choose whether to participate, the institution announces how it will divide the operating cost for every possible group of active players. This commitment matters for entry because an outsider’s incentive to join depends on the charge specified for the enlarged group. The model separates the network's role from
the institution's role.

\begin{figure}[t]
    \centering
    \includegraphics[width=0.6\textwidth]{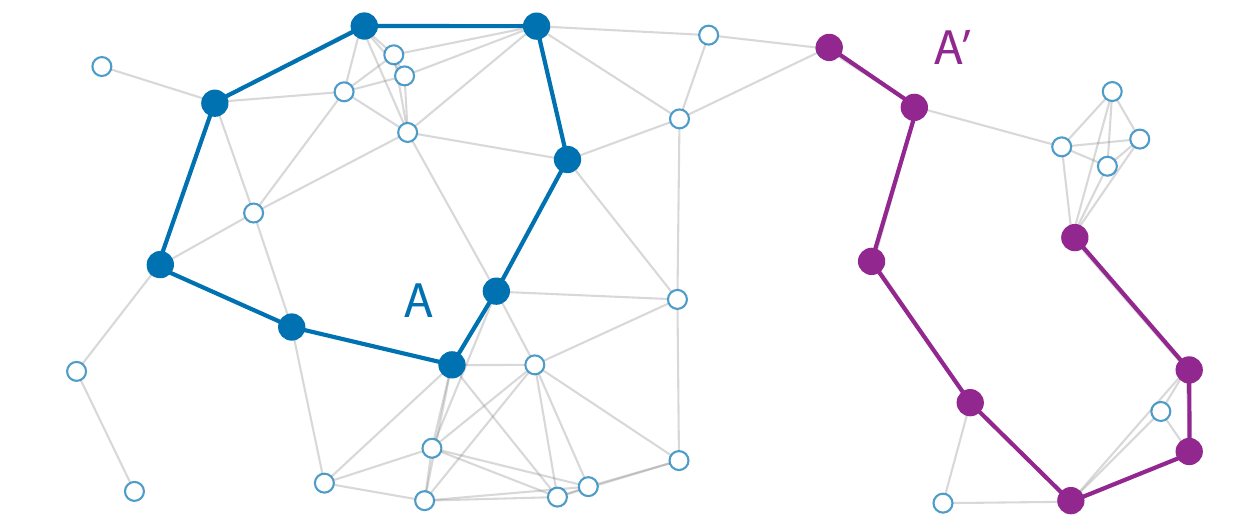}
    \caption{\textbf{Equal output, different cost-bearing capacity.}
The fixed host network contains \(40\) positions. The blue cycle and purple
path are alternative groups of eight participants, considered separately.
Both produce eight units of connected output. In the cycle, each withdrawal
destroys one unit, so the group can bear a cost equal to eight times the unit
benefit. Along the path, the withdrawal losses are
\(1,2,3,4,4,3,2,1\), so the group can bear a cost equal to twenty times the
unit benefit. At either maximal cost, every participant is indifferent between
remaining and withdrawing.}
    \label{fig:fig1a}
\end{figure}

\FloatBarrier

\subsection{Summary of contributions}
\label{subsec:contributions}

Our first result shows how network position limits voluntary payments. A
participant remains active only when her charge does not exceed the value of
the connected output her departure would destroy. These individual limits
determine the group's \emph{cost-bearing capacity}---the largest bill the group can
cover---and each participant's possible payments. She can pay nothing when the
others can cover the bill. Otherwise, their shortfall is the least she must
pay. The network therefore constrains the cost division, while the institution
usually chooses among the divisions that respect those limits. A group that
excludes players must also deter entry, so this conclusion depends on commitment
to the charges announced after an inactive player joins. At maximum capacity the
limits determine a unique division. Below this limit the institution generally retains
a choice. Lemma~\ref{lem:participation-conditions},
Theorem~\ref{thm:payment-characterization}, and
Corollary~\ref{cor:individual-payments} establish these claims.

Figure~\ref{fig:fig1a} illustrates how sharply these constraints can depend on network connections. When one output unit is
worth one, the cycle can bear a bill of eight and the path a bill of twenty,
although both produce eight units of output. At a bill of eight, every cycle
participant must pay one. With the same bill on the path, any chosen participant can be exempted of paying 
under some committed rule. The path's larger capacity comes with greater loss
after withdrawal: if each participant is equally likely to leave, expected
connected output is seven on the cycle but only five and a half on the path.
More generally, among connected networks with at least five participants, only
the path maximizes the cost-bearing capacity, and it also minimizes expected surviving output
after one equally likely withdrawal (Theorem~\ref{thm:graph-bound}).

Network structure also determines when payment becomes unavoidable. For groups
with at least three participants, whenever a committed rule can keep the group
active at a bill no greater than half of capacity, each participant can be
exempted under some rule. As the bill rises, participants whose departure would
destroy the most output are the first whom every such rule must charge. A
hub-and-spoke network forces its hub to pay just above half of capacity, whereas
a growing path postpones its first forced payer until the bill is close to
capacity.  Corollary~\ref{cor:unavoidable-threshold} gives the sharp thresholds.

The welfare result explains why these limits matter. Full participation may
uniquely maximize social welfare---the total benefit received by all players
minus the real operating cost---even though no division of the bill can keep
everyone active. The gap arises because a departing participant counts only
her own loss, although her departure reduces every player's benefit. Even the
path, which maximizes voluntary cost-bearing capacity, can finance asymptotically only one
quarter of full participation's gross social benefit---the total benefit before
the operating cost is subtracted (Proposition~\ref{prop:welfare-gap}). Fragility to random withdrawal can therefore increase the bill
participants will bear. However, even on the path, the largest such bill
remains far below the gross benefit created by full participation.

\subsection{Related work}
\label{subsec:related}

Voluntary public-good finance asks when private incentives can cover a collective cost. Classical public-good theory characterizes efficient provision \citep{Samuelson1954}, while private-provision models show how free riding can leave voluntary contributions below the efficient level \citep{BergstromBlumeVarian1986}. Threshold contribution games ask whether voluntary payments can reach the cost of a fixed project \citep{PalfreyRosenthal1984,BagnoliLipman1989}. Participation models explain why a public good may remain unprovided even when its aggregate benefits exceed its cost \citep{SaijoYamato1999,DixitOlson2000,SaijoYamato2010}, while models of self-enforcing agreements ask whether membership is stable against withdrawal and entry \citep{CarraroSiniscalco1993,Barrett1994}. Reciprocity networks show how social motives and cost sharing can jointly sustain participation \citep{DufwenbergPatel2017}. There, links represent reciprocal relationships. In our model, they represent productive connections. Our departure is that network position generates each participant’s payment limit. A participant who withdraws avoids her charge but continues to benefit from the output that survives, so she cannot be charged more than the value of the output her withdrawal destroys.

Public-good games on networks give links a different role. Links commonly
determine whose effort benefits whom, producing network-specific patterns of
free riding, specialization, and welfare
\citep{BramoulleKranton2007,GaleottiEtAl2010}. \citet{Allouch2015} studies private provision on a fixed network in which consumers benefit only from their direct neighbors’ contributions and derives equilibrium and neutrality results. \citet{ElliottGolub2019} represent marginal externalities through a weighted, directed network and characterize Lindahl contributions in terms of eigenvector centrality. Links play a different role in our model. Benefits are global and nonexcludable: every player receives the same public output. Links instead determine which active players produce connected output together and how much output each participant would destroy by withdrawing. The network is fixed, and players choose whether to participate rather than which links to form, unlike models of endogenous network formation \citep{JacksonWolinsky1996,BalaGoyal2000}.

Cost-allocation research asks how to divide a collective bill, but neighboring models differ in what they take as given. Public-project and shared-service mechanisms use agents’ reported values or demands to determine provision and payments \citep{JacksonMoulin1992,MoulinShenker1992,Moulin1994,DebRazzolini1999,MoulinShenker2001}. In network cost allocation, the infrastructure itself is the object being financed. Minimum-spanning-tree games allocate the cost of connecting users to a source \citep{Bird1976,GranotHuberman1981}, whereas strategic network-design games allow self-interested users to choose connections and share their edge costs \citep{AnshelevichEtAl2008}. The standard fixed-tree model of \citet{TijsEtAl2002} is closer to ours: agents are connected to a source through a given tree, and a cooperative cost game allocates the tree’s maintenance costs. We also take the infrastructure as given, but divide a single fixed operating bill and make participation endogenous. Players may enter or withdraw, while inactive players continue to receive whatever public benefit survives. We ask which divisions of the entire bill can sustain a proposed group against withdrawal by its members and entry by outsiders. Dividing the bill then resembles a bankruptcy or taxation problem, in which a fixed amount must be allocated under individual claims or bounds \citep{ONeill1982,AumannMaschler1985,Thomson2003,Thomson2015}. Those models take the claims or bounds as given. Here the network generates them: each participant’s position determines how much she can be charged without withdrawing.

Withdrawal responsibility connects the analysis to cooperative network games
and graph vulnerability. Cooperative network games evaluate a player through
coalition values or marginal contributions across many coalitions
\citep{Myerson1977,BachrachEtAl2013,VanDerZandenEtAl2023}. We evaluate a
participant through her unilateral departure from the realized group.
Key-player and vulnerability research likewise asks how node removal changes equilibrium activity, connectivity, or the size of the largest surviving component 
\citep{BallesterEtAl2006,Borgatti2006,BarefootEtAl1987,BaggaEtAl1992,
ChenHero2013,MattaEtAl2017}. On a tree, withdrawal responsibility is the
vertex's complementary weight, a statistic  developed through complementary-weight sequences by \citet{Shang2022}.
We therefore do not present the vertex statistic itself as new.
Theorem~\ref{thm:graph-bound} instead aggregates withdrawal responsibility across vertices, bounds the resulting sum among connected graphs of a fixed order, and characterizes the equality cases. The decisive economic difference is agency. Failure, attack, and dismantling models treat deletion as an external event or intervention 
\citep{AlbertEtAl2000,CallawayEtAl2000,SchneiderEtAl2011,
BraunsteinEtAl2016}. Here withdrawal is a participant’s voluntary response to a charge. Under a committed cost-division rule, each deletion loss serves as an individual payment limit. Together with budget balance, these limits characterize the sustainable current divisions, the group’s cost-bearing capacity, its unavoidable payers, and the voluntary-financing gap.

\section{The model}
\label{sec:model}

Let \(G=(V,E)\) be a finite, simple, undirected, connected graph with
\(N=|V|\geq3\) vertices. Each vertex \(i\in V\) is a network position occupied
by one player. Players simultaneously choose whether to be active. Every set
\(A\subseteq V\) is therefore a possible \emph{active set}. Its members are
\emph{participants} and players in \(V\setminus A\) are \emph{outsiders}. An
active set is \emph{proper} when \(A\ne V\). Our main results concern connected
active sets, although the model permits disconnected sets. A withdrawal may
fragment a connected set, and a nonadjacent outsider may enter as a separate
component.

For each active set \(A\), let \(\LCC_G(A)\) be the size of the largest
connected component of the subgraph induced by \(A\), with
\(\LCC_G(\varnothing)=0\). Every player, whether active or inactive, receives
the public benefit
\[
 b\LCC_G(A),
\]
where the \emph{unit benefit} \(b>0\) is the value of one unit of connected
output.

Every nonempty active set incurs the same fixed \emph{operating cost} or bill
\(C\geq0\). This is a real resource expenditure, not a transfer among the
modeled players, and the participants must finance it in full. A
\emph{cost-division rule} \(c\) assigns a payment \(c_i(A)\geq0\) to each
participant \(i\in A\) in every nonempty active set \(A\), with the budget balance
\[
 \sum_{i\in A}c_i(A)=C.
\]
We assume that the institution commits publicly to the cost-division rule before players choose
their activity. Payments may depend on player identity and on the realized
active set. Outsiders cannot be charged. After players choose, the payments
assigned to the resulting active set apply. A unilateral entry or withdrawal
therefore invokes the payments already specified for the new active set.
Player \(i\)'s payoff is
\[
 u_i(A;c)=
 \begin{cases}
  b\LCC_G(A)-c_i(A),&i\in A,\\
  b\LCC_G(A),&i\notin A.
 \end{cases}.
\]

An active set \(A\) is a (pure) \emph{Nash equilibrium} under \(c\) if no
player gains by changing her activity. It is strict if every player strictly
loses by doing so.  

%Our analysis below asks whether there exists such a cost-division rule under which a proposed active set is an equilibrium. It does not address how the institution selects among
%possible rules or which equilibrium is reached when a given rule supports
%several.

\section{Results}
\label{sec:nash-supportability}

\subsection{Participation incentives and Nash supportability}
\label{subsec:capacity}

The first question is when a proposed active set can be sustained against entry
and withdrawal. For a nonempty active set \(A\), participant \(i\)'s
\emph{withdrawal responsibility} is
\[
 r_i(A)=\LCC_G(A)-\LCC_G(A\setminus\{i\}),
 \qquad i\in A.
\]
It is the connected output lost when \(i\) withdraws. Withdrawal responsibility
is a counterfactual property of a network position. In general, \(r_i(A)\geq0\). If \(A\) is connected,
then \(r_i(A)\geq1\). For an outsider
\(j\notin A\), define the output added by entry as
\[
 q_j(A)=\LCC_G(A\cup\{j\})-\LCC_G(A).
\]
If \(A\) is nonempty and connected, then \(q_j(A)=1\) when \(j\) has a
neighbor in \(A\), and \(q_j(A)=0\) otherwise. For a disconnected set, entry
may join several components and add more than one unit. In every case,
\[
 r_j(A\cup\{j\})=q_j(A).
\]

These two quantities determine the payoff from every unilateral entry or
withdrawal.

\begin{lemma}[Unilateral participation conditions]
\label{lem:participation-conditions}
Let \(c\) be a cost-division rule and let \(A\ne\varnothing\). The active set
\(A\) is a Nash equilibrium under \(c\) if and only if
\begin{align}
 c_i(A)&\leq br_i(A)
 &&\text{for every }i\in A, \label{eq:withdrawal-ceilings}\\
 c_j(A\cup\{j\})&\geq bq_j(A)
 &&\text{for every }j\notin A. \label{eq:entry-charges}
\end{align}
The equilibrium is strict if and only if every applicable inequality is
strict.
\end{lemma}

\begin{proof}
Participant \(i\)'s gain from withdrawing is \(c_i(A)-br_i(A)\), whereas
outsider \(j\)'s gain from entering is
\(bq_j(A)-c_j(A\cup\{j\})\). These are all unilateral deviations because
actions are binary.
\end{proof}

The lemma gives the two output changes a monetary meaning. The value
\(br_i(A)\) is participant \(i\)'s withdrawal ceiling: a larger charge makes
withdrawal profitable. The value \(bq_j(A)\) is outsider \(j\)'s entry gain
before payment, so her post-entry charge must be at least that large. These
conditions test a given rule. The next result asks whether some committed rule
can satisfy them.

Call \(A\) \emph{Nash supportable} if it is a Nash equilibrium under some
committed cost-division rule, and \emph{strictly Nash supportable} if it is a
strict Nash equilibrium under some such rule. A rule with either property is a
\emph{supporting rule} for \(A\). Supportability is an existence property
of the target active set. It does not predict  which rule or equilibrium will
emerge. To combine the participant ceilings, define the \emph{aggregate
withdrawal responsibility} of \(A\) by
\[
 R(A)=\sum_{i\in A}r_i(A).
\]
This quantity sums separate one-player responsibilities, not the output
predicted to disappear if every participant withdrew together.

\begin{theorem}[Nash supportability]
\label{thm:payment-characterization}
Let \(A\ne\varnothing\) be connected. If \(A=V\), then \(A\) is Nash
supportable if and only if
\[
 0\leq C\leq bR(A).
\]
If \(A\subsetneq V\) is proper, then \(A\) is Nash supportable if and only if
\[
 b\leq C\leq bR(A).
\]
\end{theorem}

\begin{proof}
Suppose that a rule supports \(A\). Summing the participant inequalities in
Lemma~\ref{lem:participation-conditions} and using budget balance gives
\(C\leq bR(A)\). If \(A\) is proper, connectedness of \(G\) supplies an
outsider \(j\) adjacent to \(A\). Since \(q_j(A)=1\), the entry inequality and
\(c_j(A\cup\{j\})\leq C\) give \(C\geq b\).

Conversely, suppose that the applicable bounds on \(C\) hold. Because \(A\) is
connected, \(R(A)>0\). The proportional charges
\[
 c_i(A)=C\frac{r_i(A)}{R(A)}
\]
are nonnegative, sum to \(C\), and satisfy every withdrawal ceiling. If \(A\)
is proper, then for each outsider \(j\), prescribe
\[
 c_j(A\cup\{j\})=bq_j(A),
 \qquad
 c_i(A\cup\{j\})=\frac{C-bq_j(A)}{|A|}
 \quad(i\in A).
\]
These charges are nonnegative and budget balanced because
\(q_j(A)\in\{0,1\}\) and \(C\geq b\). The post-entry prescriptions do not
conflict because the sets
\(A\cup\{j\}\) are distinct. At every other nonempty active set, use any
nonnegative, budget-balanced division. Lemma~\ref{lem:participation-conditions}
then makes \(A\) a Nash equilibrium.

\end{proof}

The proof also identifies every cost division at \(A\) that can be part of a
supporting rule. When the bounds in
Theorem~\ref{thm:payment-characterization} hold, choose any division of the bill
at \(A\) that charges each participant no more than \(br_i(A)\). Combining this
division with the post-entry charges constructed in the proof gives a rule that
supports \(A\). Conversely, Lemma~\ref{lem:participation-conditions} shows that
no supporting rule can charge a participant more. The possible divisions at
\(A\) are therefore exactly those that respect every withdrawal ceiling.

The strict boundaries follow from the same argument. Full participation is
strictly Nash supportable exactly when
\(0\leq C<bR(V)\). A proper connected set $A$ is strictly Nash supportable exactly
when \(b<C<bR(A)\). For a proper set with \(C>b\), assigning entrant
\(j\) the charge \((C+bq_j(A))/2\) makes entry strictly unprofitable and leaves
a nonnegative remainder for the incumbents. Conversely, strict inequalities in
Lemma~\ref{lem:participation-conditions} imply the stated bounds. At
\(C=bR(A)\), every withdrawal ceiling binds. At \(C=b\), every outsider adjacent
to a proper connected set is indifferent to entry.

The upper bound \(bR(A)\) is the largest operating cost that can be divided
without inducing a participant to withdraw. A proper group must also deter an
adjacent outsider, which explains the lower bound \(C\geq b\). Full
participation has no entry constraint. One available institutional choice is
proportional responsibility division:
\begin{equation}
  c_i(A) = C r_i(A) / R(A).
  \label{eq:proportional-responsibility}
\end{equation}
The proof of Theorem \ref{thm:payment-characterization} shows this cost division can be embedded in a supporting rule whenever \(A\) is Nash supportable.

Aggregate responsibility also has an exact one-withdrawal interpretation. Let
\(A\) be connected and select the participant \(I\) that withdraws uniformly at random from
\(A\). Then
\begin{equation}
 \frac{R(A)}{|A|^2}
 +
 \mathbb E\!\left[
   \frac{\LCC_G(A\setminus\{I\})}{|A|}
 \right]
 =1.
 \label{eq:responsibility-survival}
\end{equation}
This identity follows by summing the definition of responsibility.  The second term in \eqref{eq:responsibility-survival} is the expected share of connected output that survives after withdrawal. Thus, at a fixed group size $|A|$, greater cost-bearing capacity is equivalent
to less expected surviving output after one uniformly selected withdrawal.

\subsection{Who can pay, who can be exempted, and who must pay?}
\label{subsec:incidence}

We next use the ceilings in Lemma~\ref{lem:participation-conditions} and the
supportability bounds in Theorem~\ref{thm:payment-characterization} to determine
which participants can be exempted, which ones can bear the entire cost, and who must pay under every
supporting rule.

\begin{corollary}[Exact individual payment bounds]
\label{cor:individual-payments}
Let \(A\) be connected and Nash supportable at cost \(C\). Across all
cost-division rules supporting \(A\), participant \(i\)'s payment
ranges over exactly the interval
\begin{equation}
 \max\left\{0,C-b\bigl[R(A)-r_i(A)\bigr]\right\}
 \leq c_i(A)\leq
 \min\{C,br_i(A)\}.
 \label{eq:individual-payment-bounds}
\end{equation}
\end{corollary}

\begin{proof}
Nonnegativity, budget balance, and the participant's own ceiling give the
upper bound. The other participants can pay at most
\(b[R(A)-r_i(A)]\), which gives the lower bound.

Conversely, fix a value \(y\) in the displayed interval and set
\(c_i(A)=y\). If \(A\setminus\{i\}\ne\varnothing\), assign the remainder
proportionally to the other participants' responsibilities:
\[
 c_j(A)=(C-y)\frac{r_j(A)}{R(A)-r_i(A)},
 \qquad j\ne i.
\]
The bounds on \(y\) ensure that every payment is nonnegative and no ceiling is
exceeded. If \(A=\{i\}\), budget balance already forces \(y=C\). In either
case, the extension argument following
Theorem~\ref{thm:payment-characterization} embeds these charges in a supporting
rule.
\end{proof}

Corollary~\ref{cor:individual-payments} shows that participant \(i\) can be
exempted from paying exactly when the others can cover the cost,
\(C\leq b[R(A)-r_i(A)]\). She must pay exactly when
this inequality fails. Her smallest payment is then the uncovered remainder.
She can bear the entire cost exactly when \(C\leq br_i(A)\). At capacity,
\(C=bR(A)\), every ceiling binds and the charges are uniquely determined:
\[
 c_i(A)=br_i(A),\qquad i\in A.
\]
At an intermediate positive cost below capacity, a connected active set with at
least two participants admits multiple ways to divide the cost.

To determine when at least one participant must pay, let
\[
 r_{\max}(A)=\max_{i\in A}r_i(A).
\]
Corollary~\ref{cor:individual-payments} implies that at least one participant
is an unavoidable payer if and only if
\begin{equation}
 \frac{C}{bR(A)}>
 1-\frac{r_{\max}(A)}{R(A)}.
 \label{eq:unavoidable-threshold}
\end{equation}
A participant with responsibility \(r_i(A)\) first becomes unavoidable when
\(C>b[R(A)-r_i(A)]\). A larger responsibility lowers this threshold because it
reduces the amount the others can cover. Participants with the greatest
responsibility are therefore the first whom rising costs force to pay.

Figure~\ref{fig:payment-example} makes these bounds concrete by comparing
the feasible divisions of the same operating cost in two three-participant
networks. For a fixed active set \(A\), the charges at \(A\) that can be
embedded in supporting rules form a polytope, a region defined by finitely many
linear restrictions. Because the charges sum to \(C\), this polytope lies in a
simplex whose corners assign the entire cost to one participant. In the complete graph \(K_3\), the operating
cost exhausts the network's capacity,
so every participant must pay one. In the path \(P_3\), the center must pay between
one and two, whereas either endpoint player can be exempted, though not both at once.
The path therefore identifies the center as an unavoidable payer but leaves
the institution to divide the remaining cost.

\begin{figure}[t]
    \centering
    \includegraphics[width=\textwidth]{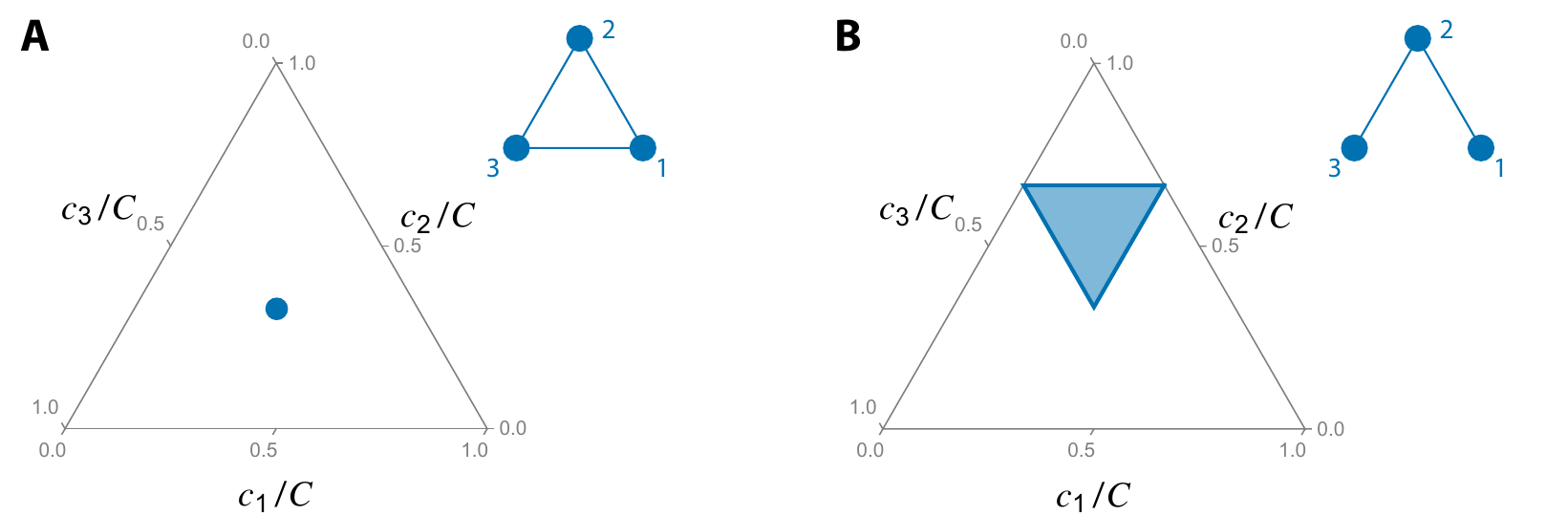}
    \caption{\textbf{Network structure changes who must pay.} All three
    participants are active, \(b=1\), and \(C=3\). Each point represents one
    division of the cost, with coordinates \(c_i(A)/C\). \textbf{A.} In
    \(K_3\), every participant has responsibility one. The operating cost
    exhausts total capacity, so \(c(A)=(1,1,1)\) is the only possible cost
    division. \textbf{B.} In \(P_3\), participant 2 is the center and has
    responsibility two. She must pay between one and two. Each endpoint can
    be exempted, but not both at once. The shaded triangle is exactly the set
    of current cost divisions that can be embedded in supporting rules.}
    \label{fig:payment-example}
\end{figure}

\FloatBarrier

\subsection{Network architecture and cost-bearing capacity}
\label{subsec:architecture}

Architecture determines both the largest cost-bearing capacity of a group and
how the corresponding payment limits are distributed across its participants.
We now identify the network architectures that minimize
and maximize this capacity and then ask when
that architecture forces someone to pay. Write \(P_m\) for the path and
\(K_{1,m-1}\) for the star on \(m\) vertices.

\begin{theorem}[Sharp aggregate-responsibility bound]
\label{thm:graph-bound}
Let \(A\) be connected with \(m=|A|\geq3\). Then
\begin{equation}
 m\leq R(A)\leq
 F_m:=\left\lfloor\frac{(m+1)^2}{4}\right\rfloor.
 \label{eq:R-bound}
\end{equation}
Equality in the upper bound holds precisely when
\[
 G[A]\cong
 \begin{cases}
 P_3,&m=3,\\
 P_4\text{ or }K_{1,3},&m=4,\\
 P_m,&m\geq5.
 \end{cases}.
\]
Equality in the lower bound holds precisely when deleting each participant
leaves the network induced by the remaining participants connected.
\end{theorem}

\ref{app:extremal} gives the proof and summarize here its main steps. Deleting edges weakly increases
withdrawal losses, so the upper-bound problem reduces to trees. A centroid
argument then establishes the bound and characterizes its equality cases. The
only exceptional size is \(m=4\): the star \(K_{1,3}\), with responsibilities
\(\{3,1,1,1\}\), and the path \(P_4\), with responsibilities
\(\{1,2,2,1\}\), both have total responsibility six. The lower bound is
attained precisely when every responsibility equals one, or equivalently,
when the network is 2-vertex-connected. Examples include every cycle \(C_m\),
complete graph \(K_m\), and complete bipartite graph \(K_{s,t}\) with
\(s,t\geq2\).

Theorem~\ref{thm:graph-bound} and Corollary~\ref{cor:individual-payments}
together give thresholds that hold for every connected architecture. Let
\[
 \lambda=\frac{C}{bR(A)}.
\]
Thus \(\lambda\) is the fraction of the active set's capacity being used.

\begin{corollary}[When architecture forces a payer]
\label{cor:unavoidable-threshold}
Let \(A\) be connected and Nash supportable at cost \(C\), with
\(m=|A|\geq3\). Then
\begin{itemize}
  \item[(i)] if \(\lambda\leq1/2\), every participant can be exempted under
  some supporting rule;
  \item[(ii)] if \(\lambda>1-1/m\), some participant must pay under every
  supporting rule.
\end{itemize}
\end{corollary}

\begin{proof}
By equation~\eqref{eq:unavoidable-threshold}, architecture forces a
participant to pay when
\[
 \lambda>1-\frac{r_{\max}(A)}{R(A)}.
\]
Every responsibility is at least one and \(r_{\max}(A)\leq m-1\), so
\[
 R(A)\geq r_{\max}(A)+(m-1)\geq2r_{\max}(A).
\]
The threshold is therefore at least \(1/2\). Equality requires one
responsibility to equal \(m-1\) and all others to equal one. Deleting the
first participant leaves only isolated vertices, so connectedness makes the
network a star.

A maximum is at least an average, so
\[
 \frac{r_{\max}(A)}{R(A)}\geq\frac1m.
\]
Equality requires all responsibilities to be equal. A leaf of any spanning
tree has responsibility one, so equality holds precisely when every
responsibility equals one, or equivalently, when every one-participant
deletion leaves the remaining network connected. The threshold is therefore
at most \(1-1/m\). The sharp-case claims now follow from
Corollary~\ref{cor:individual-payments}.
\end{proof}

Both bounds in Corollary~\ref{cor:unavoidable-threshold} are sharp. In a star,
the hub must pay if and only if \(\lambda>1/2\). If every one-participant
deletion leaves the remaining network connected, each participant can be
exempted under some supporting rule when \(\lambda\leq1-1/m\), whereas every
participant must pay when \(\lambda>1-1/m\). The result shows when architecture
begins to restrict the institution's choice. A position with high withdrawal
responsibility forces that restriction early, while evenly distributed responsibility
postpones it.

For \(m\geq5\), only the path maximizes cost-bearing capacity. On a path,
\(r_{\max}(A)=\lceil m/2\rceil\), so at least one participant becomes an
unavoidable payer only after the fraction of capacity being used exceeds
\[
 1-\frac{\lceil m/2\rceil}{F_m}.
\]
This threshold converges to one as \(m\) grows. The path therefore has the
largest cost-bearing capacity while allowing the institution to exempt from paying any
chosen participant over a wide range of costs. Specifically, whenever a
supportable cost uses no more than this fraction of capacity, each participant
can be exempted under some supporting rule.
Equation~\eqref{eq:responsibility-survival}
identifies the cost of that advantage: among connected networks of the same
size, the path also minimizes the connected output expected to survive one
uniformly selected withdrawal.

\subsection{Paying no more for greater responsibility}
\label{subsec:responsibility-order}

The preceding results leave the institution a choice among sustainable cost
divisions. Proportional division in
equation~\eqref{eq:proportional-responsibility} assigns a larger charge to a
participant with greater withdrawal responsibility whenever \(C>0\). That
ordering may suit a commercial network if a position with high withdrawal
responsibility also earns more revenue. A volunteer organization may prefer
the opposite restriction when a provider whose withdrawal destroys more
output already contributes time or effort.  We ask what the group gives up when charges may
not rise with responsibility:
\[
 r_i(A)>r_j(A)
 \quad\Longrightarrow\quad
 c_i(A)\leq c_j(A).
\]

Under this restriction, a group of \(m=|A|\) participants can finance at most at operation cost of 
\(bm\), instead of the unrestricted maximum \(bR(A)\). The restriction
therefore reduces capacity exactly when \(R(A)>m\). Every connected graph has
a participant with
responsibility one, whose withdrawal ceiling is \(b\). No more-responsible
participant may pay more, so every charge is at most \(b\).
The group can therefore cover at most \(bm\), and equal division attains this
bound when entry is also deterred.
Proposition~\ref{prop:responsibility-nonincreasing} in
\ref{app:payment-results} gives the exact weak and strict conditions.
If \(R(A)>m\), every cost-division rule supporting \(A\) at a cost above \(bm\)
must charge some more-responsible participant more than a participant with
responsibility one. On a path with \(m\geq3\), the restriction reduces capacity from
\(bF_m\) to \(bm\), a loss of
\[
b\left[
 \left\lfloor\frac{(m+1)^2}{4}\right\rfloor-m
 \right]
\]
that grows cuadratically as $m$ increases.

\subsection{Social welfare and the voluntary-financing gap}
\label{subsec:welfare}

The preceding results identify the largest operating cost that participants
can finance voluntarily. We now ask whether voluntary finance can fail even
when full participation creates more benefit than it costs.

Define \emph{social welfare} \(W\) as the sum of all \(N\) players' payoffs.
Every player receives the public benefit, and the participants' payments
finance the real operating expenditure \(C\). Hence
\[
 W(A)=Nb\LCC_G(A)-C,
\]
for every nonempty active set \(A\), while \(W(\varnothing)=0\). The division
of the cost affects participation incentives but not social welfare.

Because \(G\) is connected, full participation produces \(N\) units of output
and gross social benefit \(bN^2\). It creates more total benefit than it costs
exactly when \(C<bN^2\). By
Theorem~\ref{thm:payment-characterization}, participants can finance full
participation exactly when \(C\leq bR(V)\). The next proposition identifies
when these two thresholds conflict.

\begin{proposition}[Social welfare and the voluntary-financing gap]
\label{prop:welfare-gap}
Full participation uniquely maximizes social welfare but cannot be sustained
by any cost-division rule if and only if
\[
 bR(V)<C<bN^2.
\]
Such an operating cost exists if and only if \(R(V)<N^2\), and every connected
network satisfies this inequality.
\end{proposition}

\begin{proof}
Every proper nonempty active set has largest connected component at most
\(N-1\), whereas \(G[V]\) is connected and has size \(N\). Because every
nonempty active set incurs the same cost, full participation has strictly
greater social welfare than every proper nonempty active set. It also has
greater welfare than inactivity exactly when \(W(V)=bN^2-C>0\). Thus full
participation uniquely maximizes social welfare if and only if \(C<bN^2\).

Theorem~\ref{thm:payment-characterization} makes full participation Nash
supportable exactly when \(C\leq bR(V)\). Combining the two conditions gives
the stated interval. Such a cost exists if and only if \(R(V)<N^2\). Finally,
Theorem~\ref{thm:graph-bound} gives \(R(V)\leq F_N<N^2\) for every connected
network.
\end{proof}

The interval \(\bigl(bR(V),bN^2\bigr)\) is the voluntary-financing gap. Its
width is \(b[N^2-R(V)]\). The gap arises because a participant and society value
the same withdrawal differently. If participant \(i\) withdraws, her own
benefit falls by \(br_i(V)\), but gross social benefit falls by \(Nbr_i(V)\).
Her payment ceiling reflects only the first loss.

To compare networks of different sizes, express voluntary capacity relative to
gross social benefit:
\[
 \frac{bR(V)}{bN^2}=\frac{R(V)}{N^2}.
\]
Thus \(R(V)/N^2\) is the maximum financeable operating cost as a fraction of
full participation's gross social benefit, while \(1-R(V)/N^2\) is the gap's
width as a fraction of that benefit. Theorem~\ref{thm:graph-bound} gives
\[
 \frac{R(V)}{N^2}\leq\frac{F_N}{N^2},
 \qquad
 \lim_{N\to\infty}\frac{F_N}{N^2}=\frac14.
\]
For \(N\geq5\), only the path attains the upper bound. The maximum financeable
fraction therefore converges to one quarter.

For example, let \(C=\tfrac12 bN^2\). On a large path, capacity is about
\(\tfrac14 bN^2\), so the bill is about twice the amount participants can
finance. Full participation uniquely maximizes social welfare, yet no
cost-division rule can sustain it. More generally, fix
\(\alpha\in(1/4,1)\) and let \(C=\alpha bN^2\). For all sufficiently large
\(N\), no connected network can finance full participation, although full
participation remains the unique social-welfare maximum. Failure of voluntary
finance therefore need not mean that the operating cost exceeds gross social
benefit.

\section{Discussion and conclusion}
\label{sec:discussion}

For a connected group and an operating cost that can be supported, each participant's payment lies
between a floor and a ceiling determined by the network connections. Her ceiling is the
value of the connected output her withdrawal would destroy. Her floor is zero
if the other participants can cover the bill within their own ceilings, and
otherwise equals the amount they cannot cover. Thus, a participant can be
exempted exactly when the others can bear the entire bill, whereas she can bear
the entire bill herself only when her own ceiling is at least as large as the
bill. As the bill rises, participants with larger withdrawal losses become
impossible to exempt first. The network therefore determines each
participant's feasible payment range, while the institution generally chooses
a division within those ranges. Only at the group's cost-bearing
capacity---the sum of all payment ceilings---do all ranges collapse to a
single division.

Network architecture therefore affects finance even when group size, current output,
benefits, and the operating cost are held fixed. Additional routes can
preserve connected output after a participant withdraws, but this resilience
also improves some participants' exit options and can reduce the charges they
are willing to bear. Conversely, concentrating output dependence on particular
positions can increase cost-bearing capacity while reducing expected surviving
output after withdrawal. This tradeoff is not a prescription to build paths or
otherwise make systems fragile. Our comparison concerns a single uniformly
selected withdrawal, not targeted, repeated, or correlated failures. The broader
point is that technical resilience and financing incentives are jointly shaped
by network architecture and therefore cannot always be designed independently
in community infrastructure, conservation partnerships, volunteer
organizations, and other collective projects built from connected
contributions.

This tradeoff also constrains how operating costs can be divided. The extra
cost-bearing capacity created when some withdrawals destroy more than one unit
of output can be realized only if the institution may assign higher charges to
participants with larger withdrawal losses than to those whose withdrawal
destroys only one unit. Any organization that requires charges not to increase
with withdrawal loss must forgo this additional capacity. Withdrawal loss
therefore measures a participation constraint, not moral responsibility,
bargaining power, or entitlement to any particular division of costs. Network
structure determines which cost divisions are feasible. Social norms and
institutions determine which feasible division is chosen.

For public economics and collective-action research, our welfare result
clarifies one source of failure in voluntary finance. A participant deciding
whether to withdraw internalizes only her own loss, whereas social welfare also
counts the losses her withdrawal imposes on other beneficiaries. Voluntary
payments can therefore fail to sustain full participation even when the total
benefit it creates exceeds the operating cost. Network architecture can narrow
this financing gap, but not eliminate it: even the path, which maximizes
cost-bearing capacity, can finance asymptotically only one quarter of full
participation's gross social benefit. Instruments that weaken the link between
withdrawal and payment---such as charging inactive beneficiaries, subsidizing
the operation externally, making some benefits excludable, or enforcing
payments after withdrawal---could relax this constraint. 

Our supportability results assume that, before players choose whether to
participate, the institution can commit to a nonnegative, budget-balanced
schedule of charges for every possible active set, with charges allowed to
depend on player identity. The result is therefore an existence statement: it
shows that some such schedule can sustain a target group, but it does not
predict which schedule will be chosen, which group will form, or which
equilibrium will be selected when several exist. Full participation avoids the
entry problem. By contrast, sustaining a group that excludes some players may
depend on charges specified for the active set that would result after an
outsider enters, and those off-equilibrium charges may not be credible if the
institution can renegotiate them after entry. Equilibrium multiplicity also
remains: complete inactivity may coexist with an equilibrium that provides the
connected good. Finally, at maximum cost-bearing capacity, participants are
only weakly willing to remain active, because each is indifferent between
participating and withdrawing.

Bargaining offers a natural next step because it could determine which of the feasible cost divisions characterized here is chosen. If bargaining occurs
before participation, it could determine who pays while leaving the
network-imposed payment limits unchanged. If renegotiation is possible after
entry or withdrawal, however, the announced contingent charges may cease to be
credible. A sequential model could therefore characterize divisions that
remain stable after every participation decision. Bargaining, voting, or a
fairness criterion could then jointly determine membership and payments, with
withdrawal losses shaping participants' outside options. Such an extension
would move from characterizing feasible divisions to predicting which division,
and potentially which group, emerges.

The conclusions that the path maximizes cost-bearing capacity and that the
largest financeable share of gross social benefit approaches one quarter are
specific to our benchmark. The network is fixed, known, unweighted, and
undirected; participation is binary and simultaneous; public output equals the
size of the largest connected component; benefits are common and linear; every
nonempty active set incurs the same real operating cost; and only active
participants pay. Some parts of the analysis extend beyond these assumptions.
\ref{app:payment-results} extends the payment results to disconnected groups,
total payments by subgroups, and rules that prevent charges from increasing
with withdrawal responsibility; \ref{app:partial-information} derives payment
ceilings from local network views; and \ref{app:extensions} considers
coordinated withdrawal, heterogeneous and nonlinear benefits, and immutable
personal costs. The partial-information result is especially stark on growing
paths: neighborhoods small relative to the group justify only a vanishing
share of the cost-bearing capacity available under complete information.

Natural extensions include allowing operating costs to vary with participation,
introducing weighted or directed links, allowing participants to choose network
formation and redundancy, and studying sequential entry, withdrawal cascades,
and richer coalitional deviations. Empirically, one could estimate withdrawal
losses and test whether they predict contributions, negotiated payments, or
membership decisions.

Connections thus determine how much of a public good's value becomes privately
at stake when a participant considers leaving. They can redistribute those
stakes and expand what participants will finance, but they cannot make a
participant bear the losses her withdrawal imposes on everyone else. Fragility
can therefore strengthen voluntary finance, but even the most financeable
network cannot close the gap between the value it creates and what its
participants will pay to sustain it.

\section*{Data availability}

No data were used for the research described in the article.

\section*{Acknowledgments}

This work was partially supported by UNAM-PAPIIT IA02225.

\section*{Declarations}
\subsection*{Competing interests}

The author declares no competing interests.

\appendix

\section{Proof of the sharp aggregate-responsibility bound}
\label{app:extremal}

Let \(H=G[A]\). Thus \(H\) is connected, has vertex set \(A\), and has order
\(m\). For each connected graph \(J\) on \(A\), define
\[
 d_i(J)=m-\LCC_J(A\setminus\{i\}),
 \qquad
 D(J)=\sum_{i\in A}d_i(J).
\]
In particular, \(d_i(H)=r_i(A)\) and \(D(H)=R(A)\). Recall \(F_m\) from
Theorem~\ref{thm:graph-bound}. For every positive integer \(h\),
\[
 F_{2h}=h(h+1),
 \qquad
 F_{2h+1}=(h+1)^2.
\]

We use two lemmas.

\begin{lemma}[Spanning-tree domination]
\label{lem:spanning}
If \(T\) is a spanning tree of a connected graph \(J\), then
\(d_i(J)\leq d_i(T)\) for every \(i\in A\). Consequently,
\(D(J)\leq D(T)\).
\end{lemma}

\begin{proof}
Every component of \(T-i\) is contained in a component of \(J-i\), because
every edge of \(T-i\) is also an edge of \(J-i\). The largest component of
\(J-i\) is therefore at least as large as the largest component of \(T-i\).
Subtracting their orders from \(m\) gives \(d_i(J)\leq d_i(T)\). Summing over
\(i\) proves the second claim.
\end{proof}

\begin{lemma}[Rooted-subtree bound]
\label{lem:branch}
Let \(Q\) be a tree of order \(a\), rooted at \(v_0\). For each
\(v\in V(Q)\), let \(t(v)\) be the number of vertices in the subtree rooted at
\(v\), including \(v\). Then
\[
 \sum_{v\in V(Q)}t(v)\leq\frac{a(a+1)}2.
\]
Equality holds if and only if the tree is a path rooted at an endpoint.
\end{lemma}

\begin{proof}
We use induction on \(a\). The claim is immediate for \(a=1\). Suppose the
subtrees rooted at the children of \(v_0\) have orders
\(a_1,\ldots,a_s\), whose sum is \(a-1\). Because \(t(v_0)=a\), the induction
hypothesis gives
\[
 \sum_{v\in V(Q)}t(v)
 \leq a+\sum_{j=1}^s\frac{a_j(a_j+1)}2.
\]
Set \(g(t)=t(t+1)/2\). For positive \(t\) and \(y\),
\[
 g(t+y)-g(t)-g(y)=ty>0.
\]
Merging two nonempty child subtrees therefore strictly increases the sum of
their \(g\)-values. It follows that
\[
 \sum_{j=1}^s g(a_j)\leq g(a-1),
\]
with equality only when \(s=1\). Hence
\[
 \sum_{v\in V(Q)}t(v)\leq a+g(a-1)=g(a).
\]
Equality requires one child at the root and equality recursively in its
subtree. This is precisely a path rooted at an endpoint. Conversely, such a
path has subtree orders \(a,a-1,\ldots,1\), whose sum is \(a(a+1)/2\).
\end{proof}

\begin{proof}[Proof of Theorem~\ref{thm:graph-bound}]
Because \(H-i\) has \(m-1\) vertices, its largest component has order at most
\(m-1\). Thus \(d_i(H)\geq1\) for every \(i\), and
\[
 R(A)=D(H)\geq m.
\]
Equality holds exactly when \(H-i\) is connected for every \(i\). This proves
the lower bound and its equality statement.

We next prove the upper bound. Lemma~\ref{lem:spanning} reduces the problem to
an arbitrary tree \(T\) of order \(m\). Choose a vertex \(v_0\) such that every
component of \(T-v_0\) has at most \(m/2\) vertices. Every finite tree has such
a vertex. To find one, start at an arbitrary vertex and move into a component of the current
vertex's deletion whenever that component contains more than \(m/2\) vertices.
After a move, the component containing the previous vertex has fewer than
\(m/2\) vertices, and every other component is smaller than the one just
entered. The oversized component therefore shrinks at each step, so the
procedure terminates. This vertex is commonly called a centroid.

Let the components of \(T-v_0\) have orders \(a_1,\ldots,a_p\), and set
\[
 M=\max_j a_j.
\]
Then \(M\leq m/2\) and \(\sum_j a_j=m-1\). Root \(T\) at \(v_0\). For each
nonroot vertex \(v\), let \(t(v)\) be the order of its descendant subtree.
Because \(v\) lies in one branch at \(v_0\), we have
\(t(v)\leq M\leq m/2\). In \(T-v\), the component containing \(v_0\) has
order \(m-t(v)\), while every other component has order at most \(t(v)-1\).
The rootward component is therefore largest, and
\[
 d_v(T)=t(v)\qquad(v\ne v_0).
\]
At the root, the largest component has order \(M\), so
\(d_{v_0}(T)=m-M\). Applying Lemma~\ref{lem:branch} to each branch gives
\begin{equation}
 D(T)\leq m-M+\frac12\sum_{j=1}^p(a_j^2+a_j).
 \label{eq:centroid}
\end{equation}

Suppose first that \(m=2h+1\). Then \(M\leq h\),
\(\sum_j a_j=2h\), and
\[
 \sum_j a_j^2\leq M\sum_j a_j=2hM.
\]
Substitution into \eqref{eq:centroid} yields
\[
 D(T)\leq3h+1+(h-1)M
 \leq3h+1+h(h-1)
 =(h+1)^2=F_m.
\]
For \(h\geq2\), equality requires \(M=h\), and equality in the squared-sum
bound requires every positive \(a_j\) to equal \(h\). Since the branch orders
sum to \(2h\), there are exactly two branches of order \(h\). Equality in
Lemma~\ref{lem:branch} makes both branches paths rooted at the endpoint
adjacent to \(v_0\); hence \(T\cong P_m\). When \(h=1\), the two branches are
singletons and \(T\cong P_3\).

Now suppose \(m=2h\) for an integer \(h\geq2\). If \(M=h\), one branch has order \(h\)
and all remaining branches have total order \(h-1\). Strict superadditivity of
\(g\) bounds their total contribution by \(g(h-1)\), with equality only when
there is one remaining branch. Thus
\[
 D(T)\leq h+\frac{h(h+1)}2+\frac{h(h-1)}2
 =h(h+1)=F_m.
\]
Equality in the rooted-subtree bound makes the two branches endpoint-rooted
paths of orders \(h\) and \(h-1\), so \(T\cong P_m\).

It remains to consider \(M\leq h-1\). Using
\(\sum_j a_j^2\leq M\sum_j a_j=M(2h-1)\) in
\eqref{eq:centroid}, we obtain
\[
 D(T)
 \leq\frac{6h-1+(2h-3)M}{2}
 \leq\frac{2h^2+h+2}{2}.
\]
For \(h\geq3\), this is strictly less than \(h(h+1)=F_m\), since the
difference is \((h-2)/2>0\). If \(h=2\), then \(M\leq1\), and the three
branch orders, whose sum is three, must all equal one. Thus
\(T\cong K_{1,3}\), whose responsibilities are \(3,1,1,1\), and
\(D(T)=6=F_4\). The other equality tree of order four, obtained when \(M=h\),
is \(P_4\).

We have proved the upper bound and classified the equality trees. To verify
attainment, label the vertices of \(P_m\) consecutively. Then
\[
 d_i(P_m)=m-\max\{i-1,m-i\}=\min\{i,m+1-i\}.
\]
Consequently,
\[
 D(P_{2h})=2\sum_{i=1}^h i=h(h+1),
\]
and
\[
 D(P_{2h+1})=2\sum_{i=1}^h i+(h+1)=(h+1)^2.
\]
For \(K_{1,3}\), the center has responsibility three and each leaf has
responsibility one, giving total six.

Finally, we exclude additional edges. Suppose \(H\) attains \(F_m\), and let
\(T\) be an arbitrary spanning tree of \(H\). Spanning-tree domination and the tree
bound imply
\[
 F_m=D(H)\leq D(T)\leq F_m.
\]
Thus \(T\) is one of the equality trees just identified, and
\(d_i(H)=d_i(T)\) for every vertex \(i\).

If \(T\cong P_m\) and \(H\) contains an additional edge, that edge joins two
nonconsecutive path vertices \(v_p\) and \(v_q\). Choose \(v_s\) with
\(p<s<q\). The graph \(T-v_s\) has two nonempty components, whereas the
additional edge connects them in \(H-v_s\). Hence
\(d_{v_s}(H)=1<d_{v_s}(T)\), a contradiction. If
\(T\cong K_{1,3}\), every additional edge joins two leaves. Deleting the
center then leaves a component of at least two vertices in \(H\), whereas all
three components in \(T\) are singletons. The center's withdrawal responsibility
is therefore strictly smaller in \(H\), another contradiction. Hence \(H\)
contains no
edge outside \(T\).

Since \(D(H)=R(A)\), the upper-equality graphs are exactly \(P_3\) for \(m=3\),
\(P_4\) and \(K_{1,3}\) for \(m=4\), and \(P_m\) for \(m\geq5\).
\end{proof}

\section{Supplementary payment results}
\label{app:payment-results}

\subsection{Disconnected active sets}

The main theorem specializes to connected active sets. The same argument
also characterizes arbitrary nonempty sets. Define
\[
 R(A)=\sum_{i\in A}r_i(A)
\]
and, for a proper set, define
\[
 q_{\max}(A)=\max_{j\notin A}q_j(A),
\]
with \(q_{\max}(V)=0\).

\begin{proposition}[Disconnected supportability]
\label{prop:general-supportability}
An arbitrary nonempty active set \(A\) is Nash supportable if and only if
\[
 bq_{\max}(A)\leq C\leq bR(A).
\]
If \(A\) is proper, it is strictly Nash supportable if and only if every
participant has positive responsibility and
\[
 bq_{\max}(A)<C<bR(A).
\]
At full participation, strict support is possible if and only if every
responsibility is positive and \(0\leq C<bR(V)\).
\end{proposition}

\begin{proof}
Under any supporting rule, summing the participant ceilings gives
\(C\leq bR(A)\). For each outsider \(j\), the entry condition and
nonnegativity imply
\[
 bq_j(A)\leq c_j(A\cup\{j\})\leq C.
\]
These inequalities give the lower bound.

Conversely, assign the current payments in proportion to positive
responsibilities. If \(R(A)=0\), the upper bound forces \(C=0\), so assign
zero payments. After outsider \(j\) enters, charge her \(bq_j(A)\) and divide
the remainder among the incumbents, and complete the rule arbitrarily at all
other nonempty active sets. Lemma~\ref{lem:participation-conditions} then
supports \(A\).

For strict necessity, nonnegative charges and the strict participant
inequalities require every responsibility to be positive. Summing those
inequalities gives \(C<bR(A)\). If \(A\) is proper, the strict entry
inequalities give \(bq_{\max}(A)<C\).

Conversely, suppose the stated strict conditions hold. Proportional current
payments are strictly smaller than every withdrawal ceiling. If \(A\) is
proper, charge each entrant the entire cost. Since
\(C>bq_{\max}(A)\), every entry is strictly unprofitable. At full
participation there are no entry constraints. This proves sufficiency.
\end{proof}

\subsection{Payments by groups}

For a connected Nash-supportable set \(A\) and \(T\subseteq A\), write
\[
 r_A(T)=\sum_{i\in T}r_i(A),
 \qquad
 c_A(T)=\sum_{i\in T}c_i(A).
\]

\begin{corollary}[Exact group payment bounds]
\label{cor:group-payment-bounds}
Across all cost-division rules supporting \(A\), the payment by \(T\) ranges
over exactly the interval
\[
 \max\left\{0,C-b\bigl[R(A)-r_A(T)\bigr]\right\}
 \leq c_A(T)\leq
 \min\{C,br_A(T)\}.
\]
\end{corollary}

\begin{proof}
Nonnegativity, budget balance, and the individual ceilings give the upper
bound. The complementary group can pay at most \(b[R(A)-r_A(T)]\), which gives
the lower bound.

Conversely, let \(y\) belong to the displayed interval. Assign total payment
\(y\) within \(T\) in proportion to responsibility and assign \(C-y\) within
\(A\setminus T\) in the same way. An empty group receives total payment zero.
The resulting charges are nonnegative, budget balanced, and respect every
individual ceiling. The extension argument following
Theorem~\ref{thm:payment-characterization} therefore embeds them in a supporting
rule.
\end{proof}

Thus a group can bear the entire cost exactly when \(C\leq br_A(T)\). The
individual result in Corollary~\ref{cor:individual-payments} is the special
case \(T=\{i\}\).

\subsection{Responsibility-nonincreasing payments}

Call the charges at \(A\) responsibility-nonincreasing when
\[
 r_i(A)>r_j(A)
 \quad\Longrightarrow\quad
 c_i(A)\leq c_j(A).
\]

\begin{proposition}[Capacity under responsibility-nonincreasing payments]
\label{prop:responsibility-nonincreasing}
Let \(A\) be connected with \(m=|A|\geq2\). A supporting rule with
responsibility-nonincreasing charges at \(A\) exists precisely for
\[
 \begin{cases}
 0\leq C\leq bm,&A=V,\\
 b\leq C\leq bm,&A\subsetneq V.
 \end{cases}
\]
Strict support is possible precisely when the corresponding upper inequality
is strict and, for a proper set, the lower inequality is also strict.
\end{proposition}

\begin{proof}
Choose a leaf \(j\) of a spanning tree of \(G[A]\). Deleting \(j\) leaves the
other participants connected, so \(r_j(A)=1\) and \(c_j(A)\leq b\). A
participant with responsibility one faces the same ceiling, while a
participant with greater responsibility cannot pay more than \(j\). Hence
every participant pays at most \(b\), and budget balance gives \(C\leq bm\).
The entry lower bound for a proper set follows from
Theorem~\ref{thm:payment-characterization}.

Under strict support, the spanning-tree leaf pays strictly less than \(b\).
Every other responsibility-one participant also faces a strict ceiling below
\(b\), and the ordering prevents a more-responsible participant from paying
more than the leaf. Hence every charge is strictly less than \(b\), so
\(C<bm\). A proper target also requires \(C>b\) for strict entry deterrence.

Conversely, equal division gives \(c_i(A)=C/m\leq b\) and therefore respects
every withdrawal ceiling and the required ordering. Theorem~
\ref{thm:payment-characterization} supplies a supporting rule. Under the strict
upper bound, equal payments are strictly below every withdrawal ceiling. For a
proper set, the strict lower bound permits the strictly deterrent post-entry
charges constructed after Theorem~\ref{thm:payment-characterization}; full
participation has no entry constraint. This proves the strict statement.
\end{proof}

\section{Partial information and local responsibility certificates}
\label{app:partial-information}

The main analysis uses exact withdrawal responsibility, which requires a
complete map of active positions and links. In a decentralized system, a
participant may instead observe only nearby active positions. Remote links
matter because they determine whether branches reconnect after her withdrawal.
This appendix derives the largest payment ceiling that can be justified from
such a bounded view. It draws on work about what can be inferred from
finite-radius network views \citep{Angluin1980,Linial1992}, but does not study
communication or distributed computation.

\subsection{Local responsibility certificates}

Let \(H=G[A]\) be connected, let \(m=|A|\) be known, and fix \(k\geq1\).
Participant \(i\)'s observed neighborhood is the rooted subgraph induced by
\[
 \mathcal B_i^{(k)}
 =\{v\in A:\operatorname{dist}_{H}(i,v)\leq k\},
\]
and its frontier is
\[
 \partial\mathcal B_i^{(k)}
 =\{v\in A:\operatorname{dist}_{H}(i,v)=k\}.
\]
The observation does not reveal links leaving the frontier.

Consider the components of
\(H[\mathcal B_i^{(k)}\setminus\{i\}]\). Call a component \emph{sealed} if it
does not meet the frontier; if the observed neighborhood contains all of
\(A\), call every component sealed. Let \(S_i^{(k)}\) be the total size of
the sealed components and \(s_i^{(k)}\) the size of the largest one, with
value zero when none exists. Define
\begin{equation}
 \hat r_i^{(k)}(A)
 =m-\max\left\{m-1-S_i^{(k)},s_i^{(k)}\right\}
 =\min\left\{1+S_i^{(k)},m-s_i^{(k)}\right\}.
 \label{eq:k-hop-responsibility}
\end{equation}

\begin{proposition}[Sharp responsibility certificate]
\label{prop:k-hop-responsibility}
Let \(A\) be connected, let \(i\in A\), and let \(k\geq1\). The quantity in
\eqref{eq:k-hop-responsibility} satisfies
\[
 1\leq\hat r_i^{(k)}(A)
 \leq\hat r_i^{(k+1)}(A)
 \leq r_i(A).
\]
It is the largest lower bound on \(r_i(A)\) valid for every connected
\(m\)-vertex network consistent with participant \(i\)'s rooted \(k\)-step
observation. It equals \(r_i(A)\) when that observation contains the whole
active network.
\end{proposition}

\begin{proof}
If \(\mathcal B_i^{(k)}=A\), then every component is sealed,
\(S_i^{(k)}=m-1\), and
\[
 s_i^{(k)}=\LCC_G(A\setminus\{i\}).
\]
Equation~\eqref{eq:k-hop-responsibility} therefore gives
\(\hat r_i^{(k)}(A)=r_i(A)\).

Suppose that some vertices remain unobserved. After \(i\) withdraws, every
sealed component remains separate. All other participants number
\(m-1-S_i^{(k)}\), so no component of \(G[A\setminus\{i\}]\) has more than
\[
 \max\{m-1-S_i^{(k)},s_i^{(k)}\}
\]
vertices. This proves the lower bound.

The bound is sharp. Choose one unobserved vertex, join it to a frontier vertex
in every unsealed component, and attach all other unobserved vertices to the
resulting component. The added hub remains at distance \(k+1\), so the
completion preserves the observed induced subgraph. It creates one component
of size \(m-1-S_i^{(k)}\), while the sealed components retain their observed
sizes. Finally, every network consistent with the \((k+1)\)-step observation
is also consistent with the \(k\)-step observation. Taking minima over these
nested collections proves monotonicity.
\end{proof}

Sharpness is relative to the stated information: group size and the rooted
induced subgraph. A known host graph or other verified global information may
raise the certificate.

For example, consider the eight-player path
\[
 u\!-\!v\!-\!i\!-\!a\!-\!b\!-\!c\!-\!d\!-\!e.
\]
With \(k=3\), participant \(i\) sees the segment from \(u\) through \(c\).
After her hypothetical withdrawal, \(\{u,v\}\) is sealed, whereas
\(\{a,b,c\}\) reaches the frontier. Thus
\[
 S_i^{(3)}=s_i^{(3)}=2,
 \qquad
 \hat r_i^{(3)}(A)=3.
\]
The observation certifies the loss of \(i,u,\) and \(v\) without revealing
how the other branch continues.

\subsection{Cost-bearing capacity from local views}

Define the \emph{radius-\(k\) cost-bearing capacity} by
\[
 C^{(k)}(A)=b\sum_{i\in A}\hat r_i^{(k)}(A).
\]
This is the largest operating cost that can be covered when each charge must
be justified separately from that participant's view. Necessity follows by
summing the certified payment ceilings. If \(C\leq C^{(k)}(A)\), the charges
\[
 c_i(A)=
 \frac{C}{C^{(k)}(A)}\,b\hat r_i^{(k)}(A)
\]
cover the cost without exceeding any certified ceiling. The institution must
collect the individual certificates to divide the bill, but it need not pool
the underlying network maps. Pooling those maps could certify a larger bill.

Proposition~\ref{prop:k-hop-responsibility} implies
\[
 bm\leq C^{(k)}(A)
 \leq C^{(k+1)}(A)
 \leq bR(A).
\]
At full participation, \(C^{(k)}(A)\) is the largest Nash-supportable cost
under the certification requirement. For a proper connected group, support
also requires \(C\geq b\) and the post-entry charges used in the proof of
Theorem~\ref{thm:payment-characterization}.

For a path with \(m\geq2k+2\), the vertices at distances
\(1,\ldots,k-1\) from either endpoint have certificates
\(2,\ldots,k\); every other certificate equals one. Hence
\[
 \sum_{i\in A}\hat r_i^{(k)}(A)
 =m+2\sum_{\ell=1}^{k-1}\ell
 =m+k(k-1),
\]
and therefore
\[
 C^{(k)}(A)=b\bigl[m+k(k-1)\bigr],
 \qquad
 bR(A)=bF_m.
\]
It follows that \(C^{(k)}(A)/[bR(A)]\) tends to zero whenever \(k/m\)
tends to zero. Certifying a share of a path's full-information capacity that
stays bounded away from zero therefore requires an observation radius
proportional to group size.
The same argument applies to any verified lower bounds on withdrawal
responsibility.

\section{Additional implications and extensions}
\label{app:extensions}

\subsection{Coordinated withdrawal}

The individual payment ceilings also control coalitions whose members all
withdraw.

\begin{corollary}[Coalitional withdrawal by participants]
\label{cor:coalition-withdrawal}
Let \(A\ne\varnothing\), and let \(c\) be a cost-division rule under which \(A\)
is a Nash equilibrium. No nonempty coalition of participants can increase its
combined payoff by having all its members become inactive, even when they can
make budget-balanced transfers among themselves. If \(A\) is a strict
equilibrium under \(c\), every such joint withdrawal strictly reduces the
coalition's combined payoff.
\end{corollary}

\begin{proof}
For nonempty \(T\subseteq A\), let
\[
 \Delta_T(A)=\LCC_G(A)-\LCC_G(A\setminus T).
\]
Since \(A\setminus T\subseteq A\setminus\{i\}\), monotonicity of
largest-component size gives \(\Delta_T(A)\geq r_i(A)\) for every \(i\in T\).
The payment ceilings therefore imply
\[
 \sum_{i\in T}c_i(A)
 \leq b\sum_{i\in T}r_i(A)
 \leq |T|b\Delta_T(A).
\]
The coalition avoids the charges on the left, but each departing member loses
\(b\Delta_T(A)\) in public benefit. Its combined gain is nonpositive. At a
strict Nash equilibrium, the first inequality is strict.
\end{proof}

This result holds other players' actions and the cost-division rule fixed. It is
not a strong-Nash result and does not cover deviations that mix entry and
withdrawal.

\subsection{Nonlinear and heterogeneous public benefits}

The Nash-supportability argument does not require linear or identical benefits.
Suppose player \(i\) receives a nondecreasing benefit
\(B_i(\LCC_G(S))\) at active set \(S\). For a nonempty target \(A\), define
\[
 v_i(A)=B_i\bigl(\LCC_G(A)\bigr)
 -B_i\bigl(\LCC_G(A\setminus\{i\})\bigr),
 \qquad i\in A,
\]
and, for \(j\notin A\), define her entry value by
\[
 e_j(A)=B_j\bigl(\LCC_G(A\cup\{j\})\bigr)
 -B_j\bigl(\LCC_G(A)\bigr).
\]
Let \(e_{\max}(A)=\max_{j\notin A}e_j(A)\) for a proper target, and set
\(e_{\max}(V)=0\). The same two deviation calculations as in
Lemma~\ref{lem:participation-conditions} show that \(A\) is a Nash equilibrium
under rule \(c\) exactly when
\[
 c_i(A)\leq v_i(A)\quad(i\in A),
 \qquad
 c_j(A\cup\{j\})\geq e_j(A)\quad(j\notin A).
\]
Consequently, \(A\) is Nash supportable exactly when
\[
 e_{\max}(A)\leq C\leq\sum_{i\in A}v_i(A).
\]
Thus heterogeneity changes both boundaries: current participants determine the
upper bound, while prospective entrants determine the lower bound.

In this generalized model, the participant ceilings also control coordinated
exit. If a coalition \(T\) withdraws, monotonicity and
\(A\setminus T\subseteq A\setminus\{i\}\) imply that every departing member
loses at least \(v_i(A)\). The individual inequalities therefore continue to
prevent a coordinated exit from increasing the coalition's combined payoff.
If \(A\) is Nash supportable and every \(v_i(A)>0\), then, among the charges
at \(A\) prescribed by supporting rules, the unique division that minimizes
the largest ratio of a charge to the corresponding withdrawal value is
\[
 c_i(A)=C\frac{v_i(A)}{\sum_{\ell\in A}v_\ell(A)},
 \qquad i\in A.
\]

\subsection{Immutable personal costs}

Suppose instead that player \(i\) bears an immutable cost \(\kappa_i>0\)
whenever active, with no shared cost to reallocate. An active set \(A\) is a
pure Nash equilibrium exactly when
\[
 \kappa_i\leq br_i(A)\quad(i\in A),
 \qquad
 \kappa_j\geq bq_j(A)\quad(j\notin A).
\]
This benchmark is an exact potential game \citep{MondererShapley1996}, with
potential
\[
 b\LCC_G(A)-\sum_{i\in A}\kappa_i.
\]
Indeed, every unilateral change alters the player's payoff and this function
by the same amount.
It differs from the fixed operating cost in the main model. The capacity result
there relies on the institution's ability to divide \(C\) among participants.

\end{document}